\documentclass[draftclsnofoot,12pt,onecolumn]{IEEEtran}        
\usepackage{amsmath, amsthm, amssymb}
\usepackage{epsfig}
\usepackage{latexsym}
\usepackage{graphicx}
\usepackage{cite}
\usepackage{epsfig}
\usepackage{subfigure}
\usepackage{color}

\newtheorem{theorem}{Theorem}

\newtheorem{corollary}{Corollary}  
\newtheorem{remark}{Remark}
  \usepackage{epstopdf}

\title{Energy Detection   of Unknown Signals over  Cascaded   Fading Channels}

\author{
Paschalis~C.~Sofotasios,~Lina~Mohjazi,~Sami~Muhaidat,~Mahmoud Al-Qutayri,~and~George~K.~Karagiannidis

\thanks{P. C. Sofotasios is  with the Department of Electronics and Communications Engineering, Tampere University of Technology, 33101 Tampere, Finland and with the Department of Electrical and Computer Engineering, Aristotle University of Thessaloniki, 54124 Thessaloniki, Greece  \, (e-mail: p.sofotasios@ieee.org)  }

\thanks{L. Mohjazi is with the Centre for Communication Systems Research, Department of Electronic Engineering, University of
Surrey, GU2 7XH, Surrey, U.K. (e-mail: l.mohjazi@surrey.ac.uk)}

\thanks{S. Muhaidat is with the Department of Electrical and Computer Engineering, Khalifa University, PO Box 127788, Abu Dhabi, UAE and with
the Centre for Communication Systems Research, Department of Electronic Engineering, University of Surrey, GU2 7XH, Guildford, U.K. (e-mail:
muhaidat@ieee.org)}

\thanks{ M. Al-Qutayri is  with the Department of Electrical and Computer Engineering, Khalifa University, P.O. Box 127788, Abu
Dhabi, UAE. (e-mail:  mqutayri@kustar.ac.ae)}

\thanks{G. K. Karagiannidis is with the Department of Electrical and Computer Engineering, Khalifa University, PO Box 127788
Abu Dhabi, UAE and with the Department of Electrical and Computer Engineering, Aristotle University of Thessaloniki, 54124 Thessaloniki, Greece \, (e-mail: geokarag@ieee.org)}
}
\begin{document}

\maketitle

\begin{abstract}
Energy detection is a favorable mechanism in several applications relating to  the  identification  of deterministic unknown signals such as in radar systems and cognitive radio communications.  The present work  quantifies the detrimental effects of  cascaded multipath fading on energy detection and investigates the corresponding performance capability. A novel analytic solution is firstly derived for a generic integral that involves a product of the  Meijer $G-$function, the Marcum $Q-$function  and   arbitrary power terms. This solution is subsequently employed in the derivation of  an exact closed-form expression for the average probability of detection of unknown signals over  $N$*Rayleigh   channels. The offered results are also  extended to the case of square-law selection, which is a relatively simple and effective diversity method.  It is shown that the detection performance  is  considerably degraded by the number of cascaded channels and that these  effects can be effectively  mitigated by a non-substantial increase of diversity branches.  
 \end{abstract}


\section{Introduction}  \label{intro}

\IEEEPARstart{E}{FFECTIVE} detection of unknown signals has been  a critical research topic since  it is  typically encountered in  applications relating to RADAR,  cognitive radio and ultra-wide band systems. Energy detection (ED) has been the most widely considered detection method due to its non-coherent structure and low implementation complexity \cite{J:Alouini}  and the references therein.  Its operation is based  on the deployment of a radiometer that its  output decision is determined by comparing the received energy level with a reference energy threshold. In the context of cognitive radio, this  decision indicates the presence or absence of unknown  signals, while the overall detection capability has been typically characterized by  the corresponding probability of detection, $P_d$, and probability of false alarm, $P_f$ \cite{J:Alouini}. Since the  objective in such systems is not  information recovery, the process is solely based on the amount of its energy regardless of its other parameters, such as carrier phase information. Therefore, in the context of ED, an unknown signal is defined as a deterministic signal with an unknown form, which can be considered a sample function of a random process. However, knowledge of the spectral region to which its confined  along with the estimated SNR statistics are sufficient to allow design of suitable energy detectors  \cite{Paschalis_1}. 
\par   It is also known that fading effects create detrimental impacts on the performance of conventional and emerging communication systems \cite{Additional_1, Additional_2, Additional_3, Additional_4, Additional_5, Additional_6, Additional_7, Additional_8, Additional_9}.  Based on this and capitalizing on the fundamental contribution of \cite{J:Alouini}, numerous investigations  have been reported on the performance and behavior of ED in various practical scenarios.  Specifically,   thorough analyses  over multipath and composite   fading channels, were carried out  for both single-channel and multi-channel scenarios  in \cite{Maged, Paschalis_1, Paschalis_2, Paschalis_4, Paschalis_5, Paschalis_6, Paschalis_7, Paschalis_8} and the references therein. In the same context,  the performance over  hyper Rayleigh  conditions was analyzed in \cite{Hu}.  Nevertheless, in spite of the usefulness of the reported analyses, none of them  accounts for the case of  cascaded fading channels. The detrimental effects   in such  fading conditions are because   transmitted signals are exposed to a product of a large number of rays reflected by $N$ statistically independent scatterers. For example, it is practically demonstrated in \cite{Frolik}   that cascaded channels can accurately model the abrupt and rapid fading dynamics resulting from large mobile scatterers in inter-vehicular communications.  

Cascaded fading conditions are also encountered in   scenarios such as  multihop cooperative communications and multiple-input-multiple-output (MIMO) keyhole communications systems \cite{Karagiannidis}. In multihop relay cognitive radio networks, ED is performed by cooperative sensing, where the process of detecting the primary signal is replicated through a number of multihop sensing nodes by means  of non-regenerative relaying. Similarly, in the case of keyhole  MIMO channels, a multi-antenna secondary user employs ED to sense the spectrum of a multi-antenna primary user \cite{Zhan}. Likewise, cognitive vehicular networks (CVNs) have been proposed as an effective  method to provide adequate and robust operation in daily road traffic, and thus, vehicular nodes have to be equipped with efficient spectrum sensing capabilities \cite{Felice}.    
  
Motivated by the above, the present work quantifies the behavior  and performance of ED under  cascaded Rayleigh fading conditions for both single-channel and multi-channel scenarios. This topic was partly addressed in \cite{Ilhan} by deriving an infinite series, which was  not studied in terms of accuracy and convergence. On the contrary, in the present analysis, a novel analytic expression is firstly derived  for an infinite integral involving a product of a Meijer $G-$function, a  Marcum $Q-$function and power terms. This solution is generic and is anticipated to be useful in various analyses relating to  natural sciences and engineering. To this end, it is employed in the derivation of an exact closed-form expression for the $\overline{P}_{d}$ of ED  over $N$ cascaded Rayleigh channels. The derived  expression is subsequently employed in extending the offered results  to the case of square-law selection, which is a relatively simple and robust diversity method that exhibits improved performance.

\IEEEpubidadjcol


\section{ Channel and System Model }

\subsection{ The $N$*Rayleigh Fading  }

 It is recalled that multiplicative fading models are capable of accounting  for fading phenomena holistically \cite{Karagiannidis}.  Physically,  these models consider received signals generated by the product of a
large number of rays reflected via $N$ scatterers \cite{Karagiannidis, Frolik}. Based on this, a generic  cascaded fading model, so called $N$*Nakagami$-m$,  was proposed in \cite{Karagiannidis} that is constructed by the product of $N$ statistically independent but not necessarily identically distributed Nakagami$-m$ random variables.  When $m=1$, this model reduces to the $N$*Rayleigh distribution  with probability density function (PDF)
 
\begin{equation} \label{Cascaded}
p_{\gamma}(\gamma) =  \frac{1}{\gamma } G_{0, N}^{N, 0} \left[ \frac{\gamma}{\overline{\gamma}} \left| ^{  \, \,  \quad \, \, \, \,  \,   -     }_{\underbrace{1, 1, \cdots, 1}_{N}} \right. \right]  
\end{equation}
where $\gamma$ and $\overline{\gamma}$ denote the instantaneous and average signal-to-noise ratio (SNR), respectively and  $G(\cdot)$  is the Meijer $G$-function. 
For the special case that $N=2$, $N$*Rayleigh  reduces to the double Rayleigh fading model, which has been used to model fading conditions in realistic communication scenarios, including mobile transmitter and receivers  \cite{Frolik, Hu}. It is also noted that cascaded fading models  have been  used in modeling the keyhole channel in multi-antenna systems \cite{Karagiannidis}.   

\subsection{Energy Detection of Unknown Signals }

The detection  of unknown signals can be  modeled as a binary hypothesis-testing problem, where $H_0$ and $H_1$ correspond to the cases that a  signal  is absent or present, respectively.  Based on this, the received signal   can be expressed as \cite{J:Alouini}
 
\begin{equation}\label{eq:System_model_1}
y\left(t\right) = \left\{ {\begin{array}{*{20}{c}}
{n\left(t\right)}&{:{H_0}}\\
{hs\left(t\right) + n\left(t\right)}&{:{H_1}}
\end{array}} \right. 
\end{equation}
where $h$ and ${s\left( t \right)}$ denote the wireless channel gain and  the transmitted information signal  with average power $E_s$, respectively, and  ${n\left( t \right)}$ is the  zero-mean complex additive white Gaussian noise (AWGN) with single-sided power spectral density $N_0$. The received signal is firstly band-pass filtered at bandwidth $B$ (Hz) and the output of the filter is subsequently squared and integrated over time duration $T$. This generates the test statistic $Y$ which is typically formulated as \cite{J:Alouini}
 
 \begin{equation}\label{eq:ED_1}
Y \sim \left\{ {\begin{array}{*{20}{c}}
{\chi _{2u}^2}&{:{H_0}}\\
{\chi _{2u}^2\left( {2\gamma } \right)}&{:{H_1}}
\end{array}} \right. 
\end{equation}
where $\chi _{2u}^2$ is  a central chi-square distribution with $2u$ degrees of freedom with $u$ denoting the corresponding time-bandwidth product. Likewise, ${\chi _{2u}^2\left( {2\gamma } \right)}$ is  a non-central chi-square distribution with the same degrees of freedom and a non-centrality parameter $2\gamma$, with $\gamma  = {\left| h \right|^2} {{E_s}}/{{N_0}}$ denoting the  instantaneous SNR of the target signal. Finally, the test statistic $Y$ is compared with an energy threshold $\rm{\lambda }$, which  determines the absence or presence of the signal under test  \cite{J:Alouini}. In the case of AWGN, the  probability of false alarm and probability of detection are represented  as  follows 

\begin{equation}
{P_f} = {\rm Pr}\left( {Y > \lambda |{H_0}} \right) = G\left( {u, \frac{\lambda}{2}} \right)
\end{equation}
and 

\begin{equation}
{P_d} = {\rm Pr}\left( {Y > \lambda |{H_1}} \right) = {Q_u}\left( {\sqrt {2\gamma } ,\sqrt \lambda  } \right)
\end{equation}

%
%
respectively, where   $G\left( a, b\right)$ and ${Q_u}\left( {a,b} \right)$ are  the regularized upper incomplete gamma function and the   Marcum $Q-$function, respectively.

\section{Energy Detection  over $N$*Rayleigh Channels }
\label{analytic}

Meijer$-G$ and Marcum$-Q$ functions are particularly important in wireless communications. However, to the best of the authors' knowledge, no tabulated  solutions for integrals that involve a product of a Meijer $G-$function, a Marcum $Q-$function and arbitrary power terms exist in the literature.


\begin{theorem} 
For $t \in \mathbb{R} $,  $u, b, c \in \mathbb{R}^{+}$, $m, n, p, q \in \mathbb{N}$, and $G(\cdot ,\cdot)$ denoting the bivariate Meijer G-function, equation \eqref{solution}  is valid for the   integral: 
 
\begin{equation}  \label{integral_problem}
\mathcal{I} = \int_{0}^{\infty} x^{t-1}  Q_{u}(b \sqrt{x}, c)   \,   G_{p, q}^{m, n} \left( k x \left| ^{a_{1}, \cdots, a_{n}, \,  a_{n+1}, \cdots, a_{p} }_{b_{1}, \cdots, b_{m}, \,  b_{m+1}, \cdots, b_{q}}\right.  \right) {\rm d}x.  
\end{equation}

\begin{equation} \label{solution}
\mathcal{I} = -   \frac{2k}{b^{2}} G_{1, 0: 1, 3: p+1, q+1}^{0,1: 1, 0: m, n+1} \left( ^{0}_{-} \left|^{\quad 1/2}_{0, -u, 1/2}\right.  \left|  ^{1, t + a_{1}, \cdots, t + a_{n}, \,  t + a_{n+1}, \cdots, t + a_{p} }_{t + b_{1}, \cdots, t + b_{m}, \, 0, \, t +  b_{m+1}, \cdots, t + b_{q}} \right| \frac{c}{2}, \frac{2k}{b^{2}}    \right). 
\end{equation}
\end{theorem}

\begin{proof}
By setting $y = kx$ in \eqref{integral_problem}  it immediately  follows that 
 
\begin{equation} \label{integral_1}
\mathcal{I} = \int_{0}^{\infty}   \frac{ Q_{u}\left(\frac{b \sqrt{y}}{\sqrt{k}}, c \right) }{  y^{1 - t} k^{t}}     \,    G_{p, q}^{m, n} \left( y \left| ^{a_{1}, \cdots, a_{n}, \,  a_{n+1}, \cdots, a_{p} }_{b_{1}, \cdots, b_{m}, \,  b_{m+1}, \cdots, b_{q}}\right.  \right) {\rm d}y. 
\end{equation}
The above integral is improper; as a result, by integrating  by parts with respect to the power term and the Meijer $G-$function using \cite[eq. (07.34.21.0002.01)]{Wolfram}, one obtains 
 
\begin{equation} \label{integral_2}
\mathcal{I} =  \lim_{y \to \infty}   \frac{   f(y)    }{k^{t} } Q_{u}\left(\frac{b \sqrt{y}}{\sqrt{k}}, c \right) - \lim_{y \to 0}   \frac{   f(y)    }{k^{t} } Q_{u}\left(\frac{b \sqrt{y}}{\sqrt{k}}, c \right)   -    \frac{  1    }{k^{t} }  \int_{0}^{\infty} f(y)  \left\lbrace \frac{{\rm d}}{{\rm d} y} Q_{u}\left(\frac{b \sqrt{y}}{\sqrt{k}}, c \right) \right\rbrace  {\rm d}y 
\end{equation}
where

\begin{equation} \label{integral_3}
f(y) =  G_{p+1, q+1}^{m, n+1} \left( y \left| ^{1, t + a_{1}, \cdots, t + a_{n}, \,  t + a_{n+1}, \cdots, t + a_{p} }_{t + b_{1}, \cdots, t + b_{m}, \, 0, \, t +  b_{m+1}, \cdots, t + b_{q}}\right.  \right).  
\end{equation}
Evidently, both limits in \eqref{integral_2} approach zero, yielding
 
\begin{equation} \label{integral_4}
\mathcal{I} =   -    \frac{  1    }{k^{t} }  \int_{0}^{\infty} f(y)  \left\lbrace \frac{{\rm d}}{{\rm d} y} Q_{u}\left(\frac{b \sqrt{y}}{\sqrt{k}}, c \right) \right\rbrace  {\rm d}y.  
\end{equation}
The derivative of the Marcum $Q-$function with respect to $y$ can be determined with the aid of  \cite{Paschalis_3}. Hence, by performing the necessary change of variables, substituting in \eqref{integral_4} along with \eqref{integral_3} and carrying out basic manipulations it follows that
 
\begin{equation}  \label{integral_5}
\mathcal{I} = -   \int_{0}^{\infty}  \frac{  G_{p+1, q+1}^{m, n+1} \left( y \left| ^{1, t + a_{1}, \cdots, t + a_{n}, \,  t + a_{n+1}, \cdots, t + a_{p} }_{t + b_{1}, \cdots, t + b_{m}, \, 0, \, t +  b_{m+1}, \cdots, t + b_{q}}\right.  \right)  }{ 2  c^{-u} k^{t + 1 - \frac{ u}{2} } b^{u - 2} e^{\frac{c^{2}}{2}}   y^{\frac{u }{2}}e^{\frac{y b^{2}}{2k} } \left[  I_{u}\left( b c \sqrt{\frac{y}{k}} \right) \right]^{-1} }      {\rm d}y  
\end{equation} 
where $I_{n}(x)$ is the modified Bessel function of the first kind. 
By then expressing the $\exp(\cdot)$ and $I_{n}(\cdot)$ functions   in terms of the Meijer $G-$function, equation  \eqref{integral_6} is deduced

\begin{equation} \label{integral_6}
\mathcal{I} = 1 -   \frac{\pi b^{2} c^{2u} e^{-\frac{c^{2}}{2}}}{2^{u + 1} k^{t + 1}} \int_{0}^{\infty}  G^{1, 0}_{0,  1} \left(  \frac{b^{2} y}{2k}  \left|^{    {-} }_{0} \right. \right)   \, G^{1, 0}_{1,  3} \left( \frac{b^{2} c^{2} y}{4 k} \left|^{    \quad 1/2    }_{0, -u, 1/2} \right. \right)   G_{p+1, q+1}^{m, n+1} \left( y \left| ^{1, t + a_{1}, \cdots, t + a_{n}, \,  t + a_{n+1}, \cdots, t + a_{p} }_{t + b_{1}, \cdots, t + b_{m}, \, 0, \, t +  b_{m+1}, \cdots, t + b_{q}} \right. \right)  {\rm d}y.   
\end{equation}
Importantly, this integral can be expressed in closed-form with the aid of \cite[eq. (07.34.21.0081.01)]{Wolfram}. Thus, by performing the necessary change of variables and substituting in \eqref{integral_6} yields (\ref{solution}), which completes the proof.  
\end{proof}

\begin{remark}
The Meijer $G$-function of one and two variables
 is used increasingly in wireless communications and  can be evaluated numerically with the aid of  the  algorithm  in \cite[Table II]{Yilmaz}. 
It is also algebraically related to the respective Fox $H-$functions  and thus \eqref{solution} can be also generically represented in terms of the  Fox $H-$function for the special case that  $ \mathcal{C} = 1$. 
\end{remark}

\subsection{Average Probability of Detection in $N$*Rayleigh Fading }

As already mentioned, the generic solution in \eqref{solution} can be used in numerous applications in wireless communications. 

\begin{corollary}
For $u, \lambda, \overline{\gamma} \in \mathbb{R}^{+}$ and $N \in \mathbb{N}$, the following expression holds for the average probability of detection over $N$*Rayleigh  fading channels with $G^{-1}(\cdot, \cdot)$ denoting the inverse regularized   incomplete gamma function.

\begin{equation} \label{Average_Pd}
\overline{P}_{d} =   {\huge  G}_{2,1: 1, 3: 1, N+1}^{0,1: 1, 0: N, 1} \left( ^{\frac{u}{2}, \frac{u-1}{2}}_{\, \, \, \,  \frac{u-1}{2}} \left|^{\quad \frac{1}{2}}_{0, -u, \frac{1}{2}}\right.  \left|  ^{\, \quad  \,  1 }_{1, \cdots , 1, 0, 1} \right| - G^{-1}(u, P_f), - \frac{1}{\overline{\gamma}}    \right). 
\end{equation}
\end{corollary}

\begin{proof}
It is recalled that $\overline{P}_d$ is obtained by averaging  $P_d$ over the statistics of the corresponding fading channel, namely,  \cite{J:Alouini}
 
\begin{align}  \label{Pd_Average}
\overline{P}_{d} &=  \int_{0}^{\infty} Q_{u}(\sqrt{2 \gamma}, \sqrt{\lambda}) \, p_{\gamma}(\gamma) \,  {\rm d} \gamma  \\
&   = \int_{0}^{\infty} Q_{u}\left(\sqrt{2 \gamma}, \sqrt{2 G^{-1}(u, P_f)}\right) \, p_{\gamma}(\gamma) \,  {\rm d} \gamma.  
\end{align}
To this effect,  by substituting  \eqref{Cascaded} in \eqref{Pd_Average} it follows that 
 
\begin{equation}  \label{Pd_Cascaded}
\overline{P}_{d} =  \int_{0}^{\infty} \frac{1}{\gamma } \,  Q_{u}(\sqrt{2 \gamma}, \sqrt{\lambda})  \,   G_{0, N}^{N, 0} \left[ \frac{\gamma}{\overline{\gamma}} \left| ^{  \, \,  \quad \, \, \,   \, \, -     }_{\underbrace{1, 1, \cdots, 1}_{N}} \right.  \right]     {\rm d} \gamma.   
\end{equation}
The above expression can be evaluated  using  Theorem $1$ since the involved integral is a special case of \eqref{integral_problem}. Hence, by performing the necessary variable transformation in \eqref{solution} and substituting in \eqref{Pd_Cascaded} yields \eqref{Average_Pd}, which completes  the proof. 
\end{proof}


\subsection{Square-Law Selection (SLS)}

Square-law selection is  an effective and relatively simple diversity method \cite{J:Alouini}. Its operation is based on selecting the branch, among $L$ branches, with the maximum decision statistic, $y_{\rm SLS} = \max (y_{1}, \cdots , y_{L})$. For  ED over fading channels, the corresponding average probability of detection  is given by 
 
\begin{equation} \label{diversity}
\overline{P}_{d, {\rm SLS}} \triangleq  1 - \prod_{i = 1}^{L} \{1 - \overline{P}_{d}(\overline{\gamma}_i)\}. 
\end{equation}
\begin{corollary}
For $u, \lambda, \overline{\gamma} \in \mathbb{R}^{+}$ and $N \in \mathbb{N}$, the average probability of detection over $N$*Rayleigh  channels with square-law selection is expressed as

\begin{equation} \label{Average_Pd_SLS} 
\overline{P}_{d, {\rm SLS}} =  1 - \prod_{i = 1}^{L} \left\lbrace 1 - {\huge  G}_{2,1: 1, 3: 1, N+1}^{0,1: 1, 0: N, 1} \left( ^{\frac{u}{2}, \frac{u-1}{2}}_{\, \, \, \,  \frac{u-1}{2}} \left|^{\quad \frac{1}{2}}_{0, -u, \frac{1}{2}}\right.  \left|  ^{\, \quad  \,  1 }_{1, \cdots , 1, 0, 1} \right| - \frac{\lambda}{2}, - \frac{1}{\overline{\gamma}_{i}}    \right) \right\rbrace . 
\end{equation}
\end{corollary}  

\begin{proof}
The proof follows by substituting \eqref{Average_Pd} into \eqref{diversity}. 
\end{proof}

\begin{figure}[ht]
\centering

\subfigure[]{
   \includegraphics[width =15cm, height = 12cm] {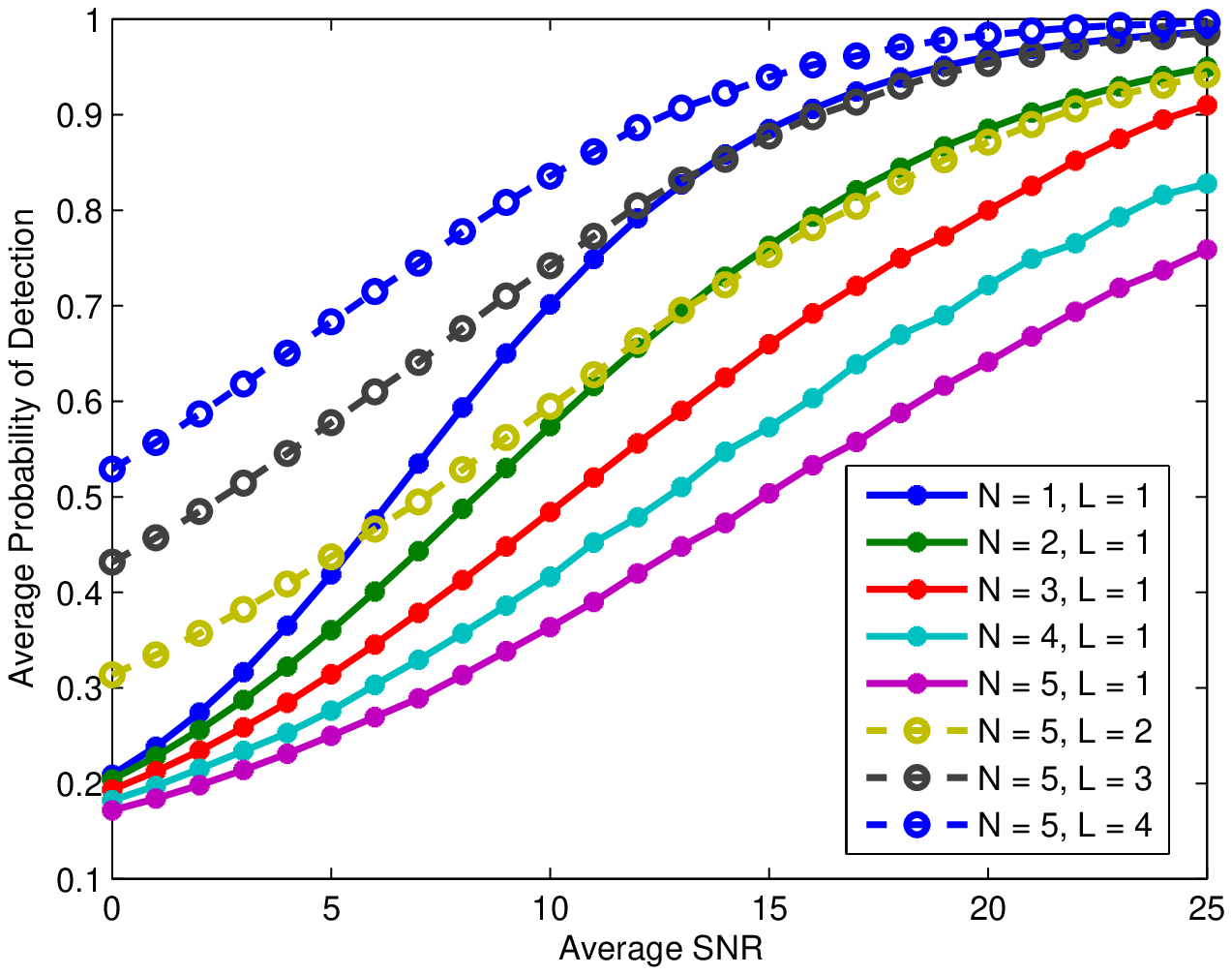}
   \label{fig:subfig1}
 }

 \subfigure[]{
   \includegraphics[width = 15cm, height = 12cm] {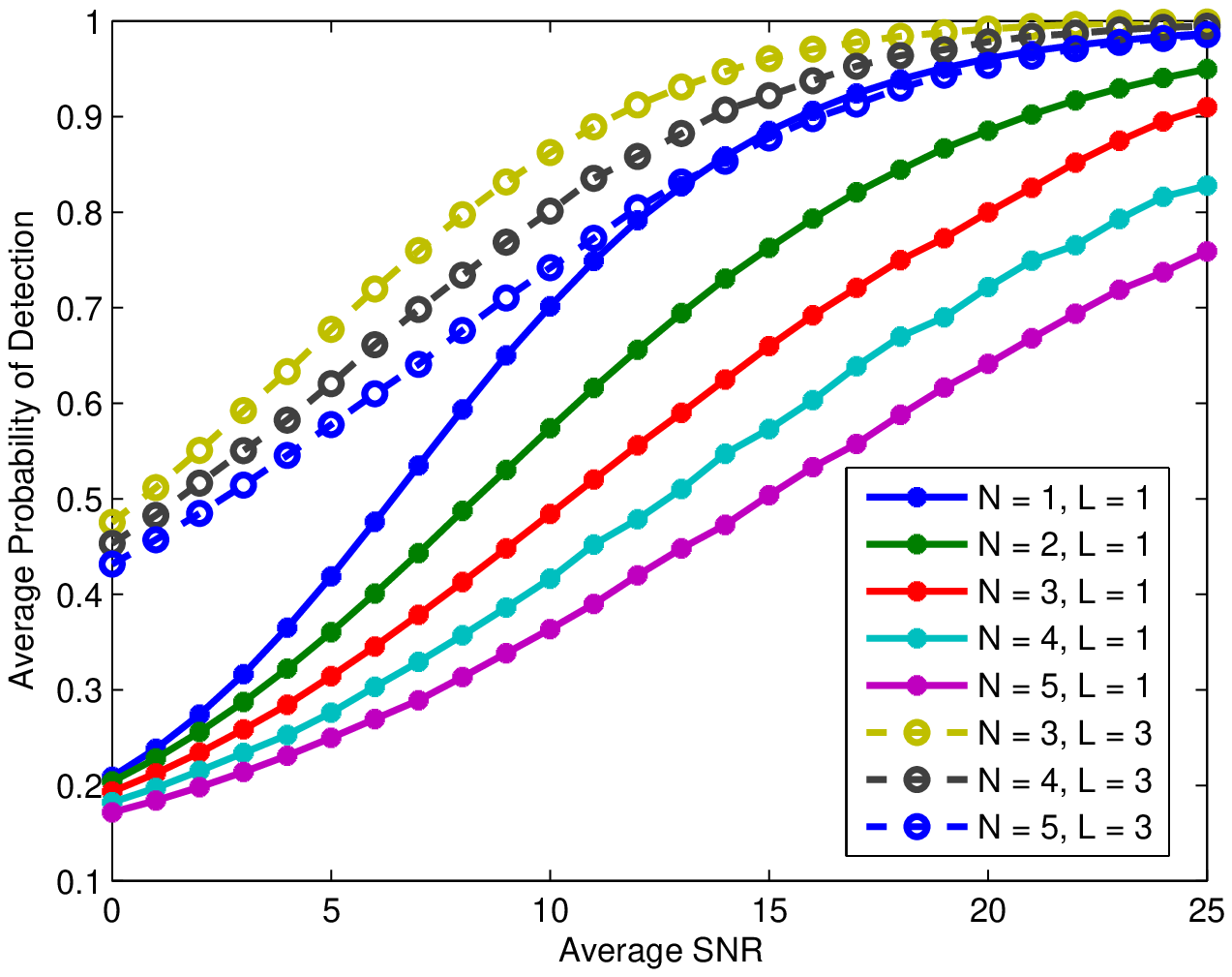}
   \label{fig:subfig2}
 }

\label{myfigure1}
\caption{$\overline{P}_d$ vs $\overline{\gamma}$ for different $N$ and $L$ with $u = 5$ and $P_f = 0.1$. }
\end{figure}

\section {Numerical Results}
\label{Results}
The offered results are used in analyzing the behavior and performance of ED based spectrum sensing  over $N$*Rayleigh fading channels for   single-channel and multi-channel cases. This enables us to quantify the corresponding detrimental effects and propose simple and efficient compensation methods. 

\begin{figure}[ht]
\centering
   \includegraphics[width =15cm, height =12cm] {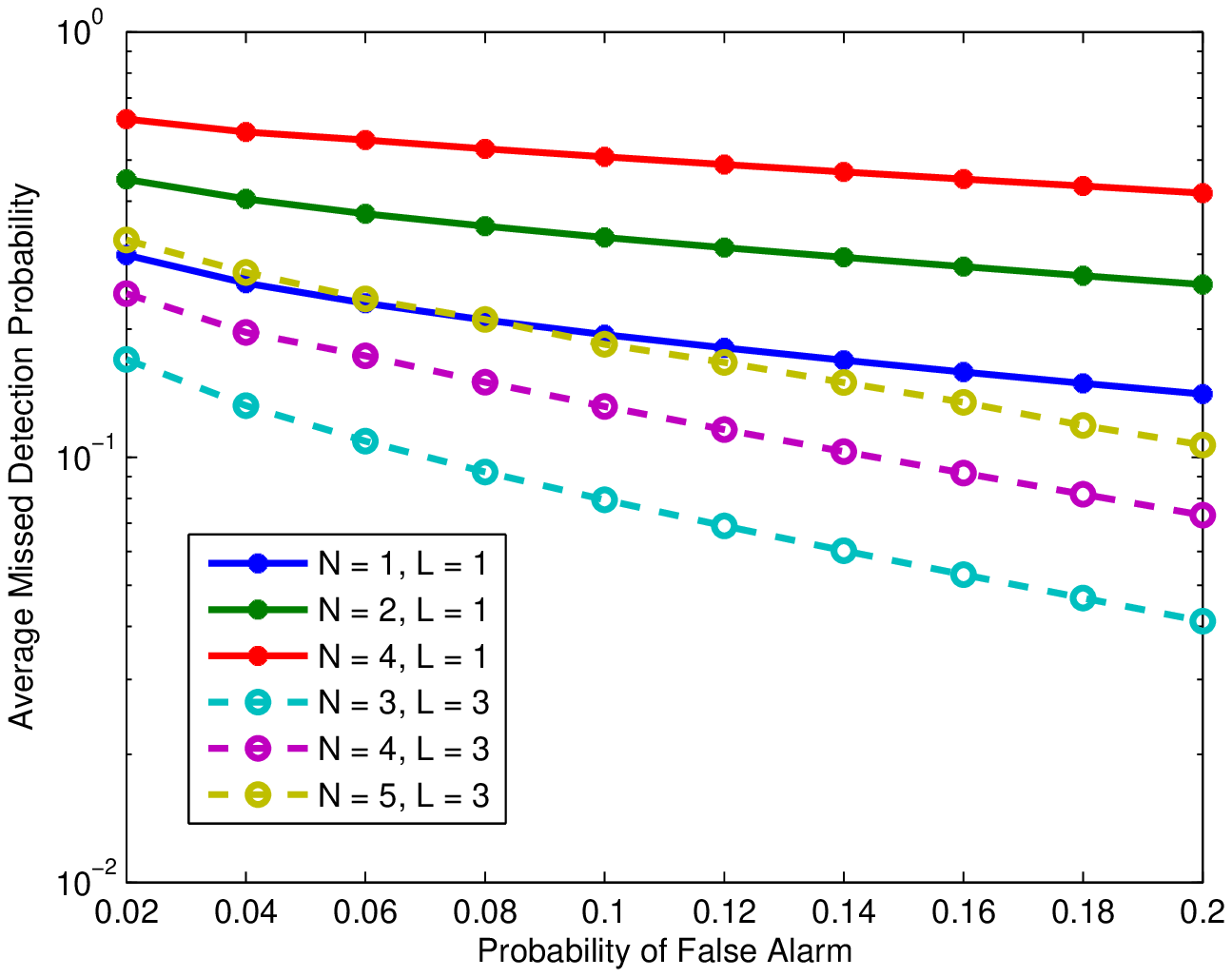}
 \label{myfigure1}
\caption{ ROC curve for different $N$ and $L$ with $u = 4$ and $\overline{\gamma} = 12$dB.}
\end{figure}

Fig.  $1$ demonstrates the average probability of detection as a function of the average SNR for     different values of $N$ and $L$ and for  $u = 5$ and $P_{f} = 0.1$, which is considered realistic for most practical applications, such as spectrum sensing.  The  value of $\overline{P}_d$ is shown to reduce  considerably as the number of cascaded channels increases. For example, for $\overline{\gamma} = 15$dB the $\overline{P}_d$ deviation between $N = 1$ and $N=3$ is   $31 \%$ whereas, that between $N=3$ and $N = 5$ for $\overline{\gamma} = 20$dB is $28 \%$.  These cases correspond to deviations that   render the ED capability inadequate since the achieved $\overline{P}_d$ is practically low even at the high SNR regime.   However, it is noticed that employing SLS  provides considerable  performance compensation even for a  small number of diversity branches. Indicatively, it is shown that  when $N=5$, the $\overline{P}_{d}$ increases substantially  for any extra added branch $L$.  This is also illustrated in the receiver operating characteristics (ROC) curve in  Fig. 2, where the value of $\overline{P}_{m}$ for $N = 5$ and $L=3$ is comparable to that for  $N = 1$ and $L=1$. As a result, the detector appears to  perform  better  when   $L\geq 4$, than in the respective single Rayleigh fading scenario. Furthermore,  the overall performance benefits of SLS are more notable in the moderate and low SNR regimes where the single channel scenario performs quite poorly.  \\
\indent 
Finally, unlike the single-channel case where increasing $N$ degrades the performance substantially, a slight degradation is observed as $N$ increases  when adopting SLS.    Overall, in  both low and high average SNR values, it is observed that considering in practice a  $3-$branch  SLS receiver can sufficiently overcome the severe  degradation by  cascaded  fading effects.

\section{Conclusion} \label{conc}

This paper quantified the effects of cascaded fading on energy detection.  A generic analytic solution was firstly derived for an integral that involves a product of a Marcum $Q-$function, a Meijer $G-$function and a power term. This solution was subsequently employed in the derivation of closed-form expressions for the average probability of detection over $N$ cascaded Rayleigh fading channels, which was then extended to the case of square-law selection.  It was shown that  the involved number of cascaded channels affect considerably the achieved performance. However, the corresponding  degradation can be effectively compensated with the aid of square-law selection since it was shown that  three or more branches are sufficient for achieving   comparable or better performance than conventional detection over   Rayleigh fading conditions.

\bibliographystyle{IEEEtran}

\end{document}